\documentclass[journal]{IEEEtran}
\usepackage{setspace}
\usepackage{fancyhdr}
\usepackage{graphicx}
\usepackage{subfigure}
\usepackage{color}
\usepackage{placeins}
\usepackage{float}
\usepackage{tabularx,colortbl}
\usepackage{amsmath}
\usepackage{amsfonts}
\usepackage{amssymb}
\usepackage{amsthm}
\usepackage[linesnumbered,boxed,ruled,commentsnumbered]{algorithm2e}
\usepackage{bm}
\usepackage{epstopdf}
\usepackage{cite}
\usepackage{diagbox}

\begin{document}

\title{Multi-cell Edge Coverage Enhancement Using Mobile UAV-Relay}
\author{
\IEEEauthorblockN{Yukuan~Ji,~\IEEEmembership{Student~Member,~IEEE}, Zhaohui~Yang,~\IEEEmembership{Member,~IEEE}, Hong~Shen,~\IEEEmembership{Member,~IEEE}, Wei~Xu,~\IEEEmembership{Senior~Member,~IEEE}, Kezhi~Wang,~\IEEEmembership{Member,~IEEE}, and~Xiaodai~Dong,~\IEEEmembership{Senior~Member,~IEEE}}\\
\thanks{
	Manuscript received October 11, 2019; revised March 6, 2020; accepted March 28, 2020. This work was supported in part by the Natural Science Foundation of Jiangsu Province for Distinguished Young Scholars under Grant BK20190012, the NSFC under grants 61871109, 61871108, and 61941115, the Fundamental Research Funds for the Central Universities under Grant 2242020K40123, and the Royal Academy of Engineering under the Distinguished Visiting Fellowship scheme. This work of X. Dong was supported in part by the NSERC of Canada under Grant 522620. \emph{(Corresponding Authors: Wei Xu, Hong Shen)}
	
	Y. Ji and H. Shen are with the National Mobile Communications Research Laboratory, Southeast University, Nanjing 210096, China (e-mail: ykji@seu.edu.cn; shhseu@seu.edu.cn).
	
	Z. Yang is with the Center for Telecommunications Research, King's College London, London WC2B 4BG, U.K. (e-mail: yang.zhaohui@kcl.ac.uk).
	
	W. Xu is with the National Mobile Communications Research Laboratory, Southeast University, Nanjing 210096, China, and also with Purple Mountain Laboratories, Nanjing 211111, China (e-mail: wxu@seu.edu.cn).
	
	K. Wang is with the Department of Computer and Information Sciences, Northumbria University, Newcastle upon Tyne NEI 8ST, U.K. (e-mail: kezhi.wang@northumbria.ac.uk).
	
	X. Dong is with the Department of Electrical and Computer Engineering, University of Victoria, Victoria, BC V8W 3P6, Canada (e-mail:xdong@ece.uvic.ca).
	
	Copyright (c) 2020 IEEE. Personal use of this material is permitted. However, permission to use this material for any other purposes must be obtained from the IEEE by sending a request to pubs-permissions@ieee.org.
}
}

\maketitle

\newtheorem{mylemma}{Lemma}
\newtheorem{mytheorem}{Theorem}
\newtheorem{mypro}{Proposition}
\newtheorem{mydis}{Discussion}
\begin{abstract}
Unmanned aerial vehicle (UAV)-assisted communication is a promising technology in future wireless communication networks. UAVs can not only help offload data traffic from ground base stations (GBSs), but also improve the quality of service of cell-edge users (CEUs). In this paper, we consider the enhancement of cell-edge communications through a mobile relay, i.e., UAV, in multi-cell networks. During each transmission period, GBSs first send data to the UAV, and then the UAV forwards its received data to CEUs according to a certain association strategy. In order to maximize the sum rate of all CEUs, we jointly optimize the UAV mobility management, including trajectory,  velocity, and acceleration, and association strategy of CEUs to the UAV, subject to minimum rate requirements of CEUs, mobility constraints of the UAV and causal buffer constraints in practice. To address the mixed-integer nonconvex problem, we transform it into two convex subproblems by applying tight bounds and relaxations. An iterative algorithm was proposed to solve the two subproblems in an alternating manner. Numerical results show that the proposed algorithm achieves higher rates of CEUs as compared with existing benchmark schemes.
\end{abstract}

\begin{IEEEkeywords}
Unmanned aerial vehicle (UAV), mobility management, trajectory optimization, user association, mobile relay.	
\end{IEEEkeywords}

\IEEEpeerreviewmaketitle
\section{Introduction}
Owing to high mobility and agility, unmanned aerial vehicle (UAV)-assisted communications have been acknowledged promising in enhancing future wireless networks \cite{Valavanis2015Handbook,Zeng2016Wireless,Chen2018Liquid,Zhao2018Integrating,Mozaffari2019Beyond}. Numerous UAV-assisted applications have emerged during the past decade, such as cargo delivery, surveillance and monitoring, with UAVs acting as different types of communication platforms including aerial base stations (BSs), aerial relays, or aerial terminals \cite{Li2019UAV,Hourani2014Optimal,Lyu2017Placement,Wang2018Power}. As aerial BSs, UAVs can provide reliable communication links for ground devices. As an aerial terminal, UAV has the degree of freedom of completing special tasks \cite{Duan2019Resource}.
In particular, UAV, acting as an aerial relay, can enlarge communication coverage and improve communication quality. Actually, the relaying technology has been widely investigated in terrestrial communications. However, most of these relays are fixed with limited mobility. Different from the static relays, UAV-enabled mobile relays offer new opportunities for performance improvement by tuning the location of UAV-relay dynamically
to best suit various specific environments, especially for latency-tolerant applications \cite{Zhan2011Wireless} and scenarios with harsh conditions \cite{Zeng2016Throughput}. Because the receivers are generally location-dispersive with mobility, the best relaying position for the receivers can vary from one to another and also from time to time. For UAV-relays with the ability of moving around, the UAV can dynamically fly near to the best position for the communication node pair \cite{Zeng2016Wireless}. Moreover,
hovering UAVs at a high altitude provide a high probability of establishing line-of-sight (LoS) links between the UAVs and ground devices \cite{Zhang2016Probabilistic}, which further leads to improved data rate and reduced latency. 

Furthermore, the distinctive characteristics of UAV make it an important technology in Internet of Things (IoT) \cite{Motlagh2016Low}. Therefore, UAV-assisted communications for 5G IoT have recently been of wide research interest. In some applications, such as agricultural surveillance, IoT devices may be deployed remotely in rural areas far from base stations. It is expensive and inconvenient to build terrestrial communication facilities to achieve the information exchange and collection for these IoT devices \cite{Motlagh2016Low,Feng2019UAV}. UAV-enabled relays help IoT devices communicate with the base station whenever necessary, which in fact expands the effective coverage of base stations. Besides, IoT devices are usually energy-limited and thus they lack the ability to communicate over a wide range.
By leveraging the mobility of UAV, it is possible to fly close to the IoT devices to communicate with them, including collecting data from devices and transmitting signals to them. In this way, the IoT devices can communicate with access points with less energy 
\cite{Mozaffari2016Unmanned,Feng2019UAV}. At the same time, the UAV can also transmit energy to energy-constrained IoT devices through radio frequency signals, which can further extend their working life \cite{Su2020UAV}. In addition, another typical application is post-disaster rescue. When cellular infrastructure is destroyed and the communication is disrupted in a sudden disaster, UAVs can be dispatched to establish temporary communication and send rescue information for IoT devices \cite{Liu2019Resource}. The IoT devices can be all kinds of human portable machine type devices, guiding humans to evacuate, avoid danger, and get rescued as soon as possible based on the rescue information.

These potential benefits of UAV relay however comes with the new challenge of three-dimensional (3D) deployment and trajectory design of UAV specifically for the communication pairs to be served \cite{He2018Joint,Pan2019Joint}. This is location based optimization in communication which is of current interest for UAV-relays \cite{Zhang2018Trajectory,Zhang2018Joint} and is most related to our current work. In particular, in order to achieve efficient and high-capacity communication, the optimal relay trajectory design of UAV requires a balance between the source-relay and relay-destination throughput. Besides, the trajectory design can greatly affect the energy efficiency of UAVs, which is a key metric especially for battery-limited UAVs \cite{Mozaffari2016Efficient,Yang2018Joint}. The authors of \cite{Yang2018Joint} proposed an iterative algorithm to minimize the sum uplink power by jointly optimizing the UAV's flight altitude, antenna beamwidth, location, transmission bandwidth and power.
In \cite{Mozaffari2017Mobile}, the deployment of multiple UAVs was optimized for collecting data from geometrically distributed IoT devices. Further considering a propulsion power consumption model, the authors of \cite{Zeng2017Energy} optimized the trajectory of UAVs aiming at maximizing energy efficiency, or equivalently lengthening the working life-time of UAVs. 

Besides the above new challenges compared to traditional relays in terms of coverage, the deployment of UAV also brings other new opportunities and technique challenges, including channel modeling \cite{Cai2017Low}, energy efficiency \cite{LinStriking}, and interference management \cite{Fouda2019Interference}. In UAV-aided communication networks, there exist both UAV-to-ground and UAV-to-UAV channels, which are quite different from well-studied traditional ground communication channels. Though the UAV-to-ground channels are usually expected as LoS links, they may also be blocked by obstacles making the reliability of communication challenging. As for UAV-to-UAV channels, they are dominated by LoS links suffering from possibly high Doppler frequency shifts. Therefore, it is necessary to measure and model these two kinds of channels more systematically \cite{Zeng2016Wireless}. Besides, UAVs suffer from limitations of size, weight, and power (SWaP), which makes the deployment and operation of energy-efficient UAVs essential for smart energy use. In the UAV-aided communication networks, there is a lack of fixed backhaul links and centralized control due to UAV's high mobility, which makes interference management more challenging than that in terrestrial communication networks. Therefore, interference management technologies especially designed for UAV communications are necessary \cite{Mei2018Uplink}.

Recently, researches have shown the potential of UAVs in expanding communication coverage or improving quality of service (QoS) of cell-edge users (CEUs). The authors of \cite{Cheng2018UAV} maximized the sum rate of a multiuser network by jointly optimizing UAV trajectory and offloading schedule among multiple cells with the UAV acting as a mobile BS. In the hybrid cellular network, the UAV-enabled BSs and ground base stations (GBSs) jointly served the ground users. In \cite{Zeng2016Throughput}, the authors investigated a  point-to-point communication system where a UAV relayed information from source to destination. Also for a point-to-point system, the authors of \cite{Pan2019Joint} proposed an algorithm to minimize the decoding error probability by jointly optimizing the time length allocation and UAV locations. Considering a layered network where a swarm of UAV was deployed to provide high QoS for IoT devices and enlarge the coverage area, the authors of \cite{Zhang2019IoT}  optimized the number of UAVs and proposed a low latency routing algorithm. In \cite{Yuan2019Joint}, a location-based beamforming scheme was proposed to enhance the security in a UAV-enabled relaying system. However, practical causal cache constraint of the UAV was not considered. In fact, the relayed information has to be buffered at the UAV before being forwarded to destinations in practice \cite{Zhao2019Caching}, which results in the causal constraint.

In this paper, we consider cell-edge performance enhancement in a multi-cell network in IoT applications
by using a UAV relay, where the UAV equipping with a cache acts as the decoding forward mobile relay to forward information from adjacent GBSs to CEUs. Main contributions of this paper are summarized as follows:
\begin{itemize}
\item We consider the scenario where a UAV-enabled mobile relay helps forward data to CEUs distributed in the joint edge coverage of multiple cells. To maximize the sum rate of all CEUs, we formulate an optimization problem by jointly optimizing the UAV mobility management, including trajectory, velocity, and acceleration, and UAV-CEU association strategy, subject to minimum rate requirements of CEUs, mobility constraints of the UAV and causal buffer constraints in practice. 
\item The cache induces an information causality constraint in practice which has rarely been considered in existing works. The UAV can only forward the data which has been successfully received during the previous time slots. This causal constraint has a very complicated form and we successfully transform it into a convex constraint by resorting to tight bounds and relaxations.
\item The original problem is a mixed-integer nonconvex problem, whose optimal solution is generally hard to be obtained. We devise an iterative algorithm to solve the original problem in an alternating manner. For the UAV mobility management optimization subproblem, we transform it into a convex problem and solve it by the well-established interior-point method. Then, we use the dual decomposition method to solve the UAV-CEU association optimization subproblem. 
\end{itemize}

The rest of this paper is organized as follows. System model and problem formulation are presented in Section \ref{sec:system_model}. In Section \ref{sec:algorithm}, we decompose the original problem into two subproblems and solve them separately. An iterative mobility management and user association (IMMUA) algorithm is proposed. The convergence and complexity are also discussed in this section. In Section \ref{sec:results}, numerical results are presented to show the effectiveness of our proposed algorithm. Finally, concluding remarks are drawn in Section \ref{sec:conclusion}.

\vspace{0.1cm}
\section{System Model And Problem Formulation}\label{sec:system_model}
\subsection{System Model}
\begin{figure}[t!]
	\setlength{\abovecaptionskip}{-1pt}
	\centering
	\includegraphics[width = 3.2in, height = 2.5in]{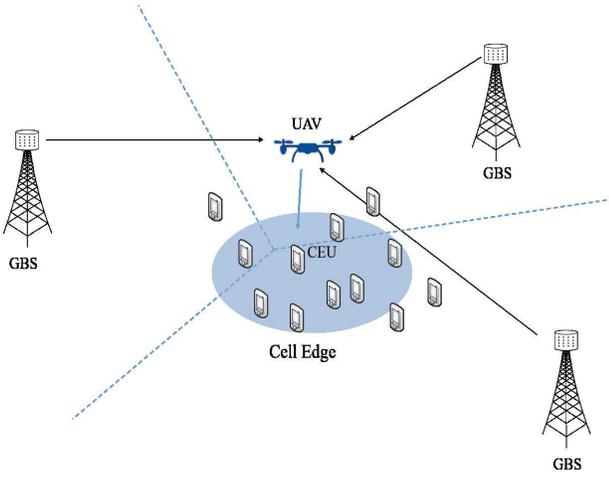}
	\caption{System model of a UAV-enabled mobile relaying network.}
	\label{Fig:system model}
\end{figure}
We consider that a mobile relay, i.e., UAV, helps serve ground CEUs distributed within the overlapped edge coverage of $ N_B $ adjacent GBSs, as depicted in Fig. \ref{Fig:system model}. In practice, it is necessary to select the appropriate type of UAV according to different application scenarios, including QoS, operating conditions, and laws \cite{Mozaffari2019A}. In general, UAVs can be roughly classified into fixed-wing and rotary-wing UAVs. The advantage of rotary-wing UAVs is that they can hover over fixed positions and fly in arbitrary direction while the disadvantages are poor mobility and limited load. In contrast, fixed-wing UAVs support higher flight speed and heavier load but need to keep flying in the air \cite{Zeng2016Wireless}. In this paper, we consider the rotary-wing UAV as a mobile relay, allowing the UAV to hover above the optimal positions and rotate by an arbitrary angle \cite{Filippone2006Flight}. Under this consideration, the constraint on the rotation of UAV can be released. The system model of UAV-assisted communication is considered in IoT applications. The CEUs can be all kinds of IoT devices, such as sensors or actuators in smart agriculture and robots in machine-to-machine (M2M) scenarios. The system model can also be applied to other scenarios via slight modifications, such as a post-disaster communication scenario.

In this paper, we focus on cell-edge users rather than cell-center users to emphasize the great potential of UAVs in improving QoS of edge users. We assume that each GBS is equipped with $ L $ antennas, and the UAV as well as each user respectively has a single antenna. Denote $ \mathcal{K}=\{1,\cdots,K\} $ as the set of CEUs. Each CEU is served via the UAV relay in order to enhance the cell-edge performance under a minimum rate constraint. During each transmission period of $ T $ seconds, the adjacent GBSs first send data to UAV, and then the UAV forwards the data to CEUs. The period length can be selected according to the application scenario and mission type of the UAV. The UAV receives and forwards data in frequency division duplexing (FDD) mode. 

3D Cartesian coordinate system is considered, in which all GBSs as well as CEUs have zero altitude and the UAV flies at a fixed altitude $ H $ in meters. The GBSs and CEUs have fixed horizontal locations denoted respectively by $ \bm{b}_m=(x_m,y_m) $ for the $ m $th GBS and $ \bm{e}_k=(x_k,y_k) $ for CEU $ k\in\mathcal{K} $. The locations of GBSs and CEUs are assumed known to the UAV. Without loss of generality, each transmission period $ T $ is splitted into $ N $ equal-length time slots, with $ \delta_t=\frac{T}{n} $ denoting the elementary slot length. $ N $ can be chosen sufficiently large in order to guarantee an approximately constant UAV location within each time slot, resulting in a sufficiently small value of $ \delta_t $. At each time slot $ n\in \{1,\cdots,N\} \triangleq \mathcal N $, the horizontal coordinate of the UAV is expressed as $ \bm{u}[n]=(x[n],y[n]) $. Therefore, the UAV's trajectory can be expressed approximately by the $ N $-length sequence $ \{(x[n],y[n])\}_{n=1}^{N} $. Similarly, the UAV's velocity and acceleration can be denoted as $ \{\bm{v}[n]\}_{n=1}^{N} $ and $ \{\bm{a}[n]\}_{n=1}^{N} $, respectively.

We assume that the UAV must return to the pre-specified starting point after each transmission period, which is denoted as $ \bm{u}_0=(x_0,y_0) $. For sufficiently small $ \delta_t $, the mobility constrains of the UAV, including starting point, terminal point, speed constraint, and acceleration constraint can be expressed as
\begin{gather}
	\bm{u}[1]=\bm{u}[N]=\bm{u}_0,\\
	\|\bm{u}[n+1]-\bm{u}[n]\|\leq V_{\rm{max}}\delta_t,\ n=1,\cdots,N-1,\\
    \bm{v}[n+1]=\bm{v}[n]+\bm{a}[n]\delta_t, \ n=1,\cdots,N-1, \label{eq:velocity} \\
    \bm{u}[n+1]\!=\!\bm{u}[n]\!+\!\bm{v}[n]\delta_t\!+\!\frac{1}{2}\bm{a}[n]\delta_t^2, \ n=1,\cdots,N-1, \label{eq:acceleration} \\
	\|\bm{v}[n]\|\leq V_{\rm{max}}, \ n=1,\cdots,N,\\
	\|\bm{a}[n]\|\leq a_{\rm{max}}, \ n=1,\cdots,N,
\end{gather}
where $ \|\cdot\| $ represents the Euclidean norm of a vector, $ V_{\rm{max}} $ and $ a_{\rm{max}} $ denote the maximum speed and maximum acceleration, respectively, and $ V_{\rm{max}}\delta_t $ is the maximum displacement in each time slot. \eqref{eq:velocity} and \eqref{eq:acceleration} are obtained according to kinematics formulas.

According to the above coordinate representation, the link distance between the GBS $ m $ and the UAV at the $ n $th time slot can be expressed as
\begin{equation}
	d_{mu}[n]=\sqrt{H^2+\|\bm{u}[n]-\bm{b}_m\|^2}.
\end{equation}
Similarly, the link distance between the CEU $ k $ and the UAV at the $ n $th time slot can be expressed as
\begin{equation}
	d_{uk}[n]=\sqrt{H^2+\|\bm{u}[n]-\bm{e}_k\|^2}.
\end{equation}

Considering high altitude of UAV, air-to-ground channels between the UAV and GBSs are dominated by line-of-sight (LoS) links. By applying the free-space path-loss model \cite{Goldsmith2005Wireless}, the channel power gain from GBS $ m $ to the UAV during time slot $ n $ is
\begin{equation}
h_{mu}[n]=\alpha_0d_{mu}^{-2}[n]=\frac{\alpha_0}{H^2+\|\bm{u}[n]-\bm{b}_m\|^2},
\end{equation}
where $ \alpha_0 $ denotes the reference channel power at 1~m.
Similarly, the channel power gain from the UAV to CEU $ k $ during time slot $ n $ is 
\begin{equation}
h_{uk}[n]=\alpha_0d_{uk}^{-2}[n]=\frac{\alpha_0}{H^2+\|\bm{u}[n]-\bm{e}_k\|^2}.
\end{equation}

The $ N_B $ adjacent GBSs use the maximum ratio transmission (MRT) strategy to transmit data to the UAV. The MRT precodes the transmitted signal by using the weights proportional to the corresponding channel coefficients, which can maximize the signal-to-noise ratio (SNR) of the received signal. In the $ n $th time slot, the signal received by the UAV from $ N_B $ GBSs is
\begin{equation} \label{eq:y_m}
y[n] = \sum_{m=1}^{N_B}\sqrt{h_{mu}[n]P_B}\bm{g}_m^H \bm{w}_m s + z,\ m=1,\cdots,N_B,
\end{equation}
where $ P_B $ represents the transmitting power of GBS, $ \bm{g}_m $ accounts for the small-scale channel fading from GBS $ m $ to the UAV, $ \bm{w}_m $ is the beamforming vector, $ s $ is the transmit signal with unit power, and $ z $ is the additive white Gaussian noise (AWGN) with variance $ \sigma^2 $. Assume that full channel state information is known to the GBSs.

Considering the $ N_B $ GBSs as a large GBS with $ N_BL $ antennas in total, the overall channel can be modeled as
\begin{equation}
\bm{g}=\left(\sqrt{h_{1u}}\bm{g}_1,\cdots,\sqrt{h_{mu}}\bm{g}_m,\cdots,\sqrt{h_{N_Bu}}\bm{g}_{N_B}\right)^T.
\end{equation}
By applying the MRT beamforming $ \bm{w}=\frac{\bm{g}}{\|\bm{g}\|} $, the signal received by the UAV in \eqref{eq:y_m} can be rewritten as
\begin{equation}
y[n]=\sqrt{P_B} \|\bm{g}\| s+z.
\end{equation}
In this way, the received SNR can be maximized as
\begin{equation}
\text{SNR}=\frac{P_B \|\bm{g}\|^2}{\sigma^2}=\sum_{m =1}^{N_B}\frac{P_B h_{mu}\|\bm{g}_m\|^2}{\sigma^2}.
\end{equation}

By using the Shannon formula,
the data rate of the UAV in the $ n $th time slot can be evaluated as
\begin{equation}\label{eq:rate_r}
R_{U_r}[n]=\log_2 \left(1+\sum_{m=1}^{N_B} \frac{P_B h_{mu}[n] \|\bm{g}_m\|^2}{\sigma^2} \right).
\end{equation}

Then, the UAV forwards the received data to its associated CEU in each time slot. Assume that the UAV serves at most one CEU during each time slot. Let $ \rho_{k,n}=1 $ indicate that the $ k $th CEU associates with the UAV for reception in the $ n $th time slot and otherwise $ \rho_{k,n}=0 $. As a result, the average rate of CEU $ k $ within $ T $ equals
\begin{equation}\label{eq:rate_k}
R_E[k]=\frac{1}{N}\sum_{n=1}^{N}\rho_{k,n}\log_2\left(1+\frac{P_U h_{uk}[n]}{\sigma^2}\right),
\end{equation}
where $ P_U $ is the transmitting power of the UAV. 

From the perspective of UAV, the transmission rate from the UAV to its associated CEUs in the $ n $th time slot can be obtained as
\begin{equation}\label{eq:rate_t}
R_{U_t}[n]=\sum_{k=1}^{K}\rho_{k,n}\log_2\left(1+\frac{P_U h_{uk}[n]}{\sigma^2}\right).
\end{equation}

For the UAV with a sufficiently large buffer, without loss of generality, the processing time at the UAV is set as one time slot. The data received in the $ n $th time slot can be forwarded in the next time slot. So the UAV has no data to forward in the first time slot and the GBSs should not transmit any data to the UAV in the last time slot. Therefore, for $ n=1 $ and $ n=N $, we have $ R_{U_t}[1]=R_{U_r}[N]=0 $ and $ \rho_{k,1}=0 $. Considering causality in practice and from \eqref{eq:rate_r} and \eqref{eq:rate_t}, we can express the causal buffer constraint as
\begin{equation}\label{eq:causality}
\sum_{i=2}^{n}R_{U_t}[i]\leq \sum_{i=1}^{n-1}R_{U_r}[i], \ n=2,\cdots,N.
\end{equation}
It guarantees that the UAV in each time slot $ n $ can only forward the data that has been successfully received during the previous time slots.

\subsection{Problem Formulation}
We define the UAV's trajectory $ \bm{U}\triangleq \{\bm{u}[n],n=2,\cdots,N-1\} $, the velocity $ \bm{V}\triangleq \{\bm{v}[n],n=1,\cdots,N\} $, the acceleration $ \bm{A}\triangleq \{\bm{a}[n],n=1,\cdots,N\} $, and the UAV-CEU association strategy $ \bm{P}\triangleq \{\rho_{k,n},\forall k\in \mathcal{K},n=2,\cdots,N\} $. Our objective is to maximize the sum rate of all the CEUs by jointly optimizing $ \bm{U} $, $ \bm{V} $, $ \bm{A} $, and $ \bm{P} $ in the transmission period $ T $, subject to the minimum rate requirements of CEUs, mobility constraints of the UAV and causal buffer constraints in practice. Then, we can formulate the joint optimization problem as
\begin{subequations}\label{eq:pro_U_P}
	\begin{align}
	\mathop{\max}_{\bm{U}, \bm{V},\bm{A},\bm{P}}\quad\!\!
	&\sum_{k\in \mathcal{K}}R_E[k]\\
	\textrm{s.t.}\quad\!\!
	&R_E[k]\geq R_0, \ \forall k\in \mathcal{K}, \\
	&\sum_{k\in \mathcal{K}}\rho_{k,n}\leq 1, \ n=2,\cdots,N,\label{eq:association_2}\\
	&\rho_{k,n}=\{0,1\}, \ \forall k\in \mathcal{K}, \ n=2,\cdots,N,\\
	&\sum_{i=2}^{n}R_{U_t}[i]\leq \sum_{i=1}^{n-1}R_{U_r}[i], \ n=2,\cdots,N,	\\
	&\bm{u}[1]=\bm{u}[N]=\bm{u}_0,\\
	&\|\bm{u}[n\!+\!1]\!-\!\bm{u}[n]\|\!\leq\!V_{\rm{max}} \delta_t,\ n\!=\!1,\!\cdots\!,N\!\!-\!\!1,\!\!\!\!\!\!\!\!\!\\
    &\bm{v}[n\!+\!1]\!=\!\bm{v}[n]\!+\!\bm{a}[n]\delta_t, \ n=1,\cdots,N\!\!-\!\!1, \\
    &\bm{u}[n\!\!+\!\!1]\!=\!\bm{u}[n]\!\!+\!\!\bm{v}[n]\delta_t\!\!+\!\!\frac{1}{2}\bm{a}[n]\delta_t^2, n\!=\!1,\!\cdots\!,N\!\!-\!\!1, \\
    &\|\bm{v}[n]\|\leq V_{\rm{max}}, \ n=1,\cdots,N, \\
    &\|\bm{a}[n]\|\leq a_{\rm{max}}, \ n=1,\cdots,N,
	\end{align}	
\end{subequations}
where $ R_0 $ is the minimum rate requirement of each CEU, which guarantees the QoS of CEUs. Constraints (19c) and (19d) ensure that the UAV serves at most one CEU in each time slot. Constraint (19e) is the information causality constraint in practice. Constraints (19f)-(19k) are mobility constraints of the UAV in terms of initial location, terminal location, speed constraint, and acceleration constraint.  

\section{Joint Mobility Management and Association Optimization}\label{sec:algorithm}
We observe that the optimization problem in \eqref{eq:pro_U_P} is a mixed-integer nonconvex problem which is generally NP-hard. It is extremely challenging to obtain its optimal solution efficiently. As can be seen from \eqref{eq:pro_U_P}, the UAV's trajectory $ \bm{U} $, the velocity $ \bm{V} $, and the acceleration $ \bm{A} $ are coupled with each other, which can be collectively referred to as mobility management. We thus decompose the original problem into two subproblems, i.e., UAV mobility management and UAV-CEU association optimization. In order to make each subproblem tractable, we resort to employing bounding and relaxation to tackle the nonconvex objective function and constraints. In this way, the original nonconvex problem can be transformed into two convex subproblems and a procedure of alternating optimization is then applied. 

It is worth mentioning that the most difficult constraint of the original problem is (19e), i.e., the information causality constraints, because the denominators on both sides of the inequality contain optimization variables. As far as we know, no methods in the existing literature has tried to deal with the constraint of this kind.

\subsection{UAV Mobility Management}\label{sec:trajectory optimization}
In order to solve the optimization problem \eqref{eq:pro_U_P}, we first optimize the UAV mobility strategy, including UAV trajectory $ \bm{U} $, the velocity $ \bm{V} $, and the acceleration $ \bm{A} $, with temporarily fixed UAV-CEU association strategy $ \bm{P} $. This subproblem can be expressed as
\vspace{0.1cm}
\begin{subequations}\label{eq:pro_U}
	\begin{align}
	\mathop{\max}_{\bm{U}, \bm{V}, \bm{A}}\quad\!\!
	&\sum_{k \in \mathcal{K}}R_E[k]\\
	\textrm{s.t.}\quad\!\!
	&R_E[k]\geq R_0, \ \forall k\in \mathcal{K}, \\
	&\sum_{i=2}^{n}R_{U_t}[i]\leq \sum_{i=1}^{n-1}R_{U_r}[i], \ n=2,\cdots,N,\\
	&\bm{u}[1]=\bm{u}[N]=\bm{u}_0,\\
	&\|\bm{u}[n+1]\!-\!\bm{u}[n]\|\!\leq\!V_{\rm{max}} \delta_t,\ n\!=\!1,\cdots\!,N\!-\!1,\\
	&\bm{v}[n+1]=\bm{v}[n]+\bm{a}[n]\delta_t, \ n=1,\cdots,N-1,\\
	&\bm{u}[n\!\!+\!\!1]\!=\!\bm{u}[n]\!\!+\!\!\bm{v}[n]\delta_t\!\!+\!\!\frac{1}{2}\bm{a}[n]\delta_t^2, n\!=\!1,\!\cdots\!,N\!\!-\!\!1, \\
	&\|\bm{v}[n]\|\leq V_{\rm{max}}, \ n=1,\cdots,N, \\
	&\|\bm{a}[n]\|\leq a_{\rm{max}}, \ n=1,\cdots,N.
	\end{align}
\end{subequations}

Problem \eqref{eq:pro_U} is, however, still nonconvex due to its nonconvex objective function and constraints (20b) and (20c). The remaining constraints are relatively easy to solve. In particular, (20d), (20f), and (20g) are linear functions of the variables $ \bm{u}[n] $, $ \bm{v}[n] $, and $ \bm{a}[n] $. (20e), (20h), and (20i) are convex with respect to $ \bm{u}[n] $, $ \bm{v}[n] $, and $ \bm{a}[n] $, respectively. We try to transform the nonconvex objective function and constraints into convex objective function and constraints via relaxation. As mentioned above, the major challenge is to deal with the causality constraints (20c). To obtain a tractable form of (20c), we need to determine a convex upper bound of the left hand side (LHS) of (20c) and a concave lower bound of the right hand side (RHS) of (20c). Without loss of generality, we consider the $ (l+1) $th iteration, given the UAV-CEU association strategy obtained at the $ l $th iteration as $ \bm{P}^{(l)} $. 

Firstly, $ R_{U_r}[n] $ can be bounded by
\begin{eqnarray}\label{eq:Rrlb}
&&\!\!\!\!\!\!\!\!\!\!\!\!\!\!\!R_{U_r}[n]=\log_2 \left(1+\sum_{m=1}^{N_B} \frac{P_B h_{mu}[n] \|\bm{g}_m\|^2}{\sigma^2} \right) \nonumber\\
&&\!\!\!\!\!\!\!\!\!\!\!\!\overset{(a)}\geq \frac{1}{N_B}\sum_{m=1}^{N_B}\log_2\left(1+\frac{N_B P_B\alpha_0 \|\bm{g}_m\|^2}{\sigma^2(H^2+\|\bm{u}[n]-\bm{b}_m\|^2)}\right)\nonumber\\
&&\!\!\!\!\!\!\!\!\!\!\!\!\overset{(b)}\geq \sum_{m=1}^{N_B}\!\left(-a_r ^{(l)}[n](\|\bm{u}[n]\!\!-\!\!\bm{b}_m\|^2\!\!-\!\!\|\bm{u}^{(l)}[n]\!-\!\bm{b}_m\|^2)\!+\!b_r ^{(l)}[n]\right) \nonumber\\
&&\!\!\!\!\!\!\!\!\!\!\!\!\triangleq \underline{R}_{U_r}[n],
\end{eqnarray}
where inequality (a) follows from the Jensen's Inequality, inequality (b) is due to the facts that $ f(x)=\log \left(1+\frac{1}{x}\right) $ is convex with respect to $ x $ and the first-order Taylor approximation is a global under-estimator of convex functions \cite{Boyd2004Convex}, and
\begin{equation}\label{eq:a}
	a_r ^{(l)}[n]=\frac{\frac{P_B \alpha_0 \sigma^2 \|\bm{g}_m\|^2}{[\sigma^2(H^2+\|\bm{u}^{(l)}[n]-\bm{b}_m\|^2)]^2}\log_2 e}{1+\frac{N_B P_B \alpha_0 \|\bm{g}_m\|^2}{\sigma^2(H^2+\|\bm{u}^{(l)}[n]-\bm{b}_m\|^2)}}\geq 0,
\end{equation} \vspace{0.2cm}
\begin{equation}\label{eq:b}
	b_r ^{(l)}[n]=\!\frac{1}{N_B}\log_2\left(1\!+\!\frac{N_B P_{B}\alpha_0\|\bm{g}_m\|^2}{\sigma^2(H^2+\|\bm{u}^{(l)}[n]-\bm{b}_m\|^2)}\right).
\end{equation}
\vspace{0.1cm}

Since the coefficient $ a_r ^{(l)}[n] $ is a nonnegative value, the lower bound of $ R_{U_r}[n] $, i.e., $ \underline{R}_{U_r}[n] $, in \eqref{eq:Rrlb} is concave with respect to $ \bm{u}[n] $. Thus far, we obtain a concave lower bound of the RHS of (20c).

\begin{figure}[t!]
	\setlength{\abovecaptionskip}{-5pt}
	\centering
	\includegraphics[width = 3.7in]{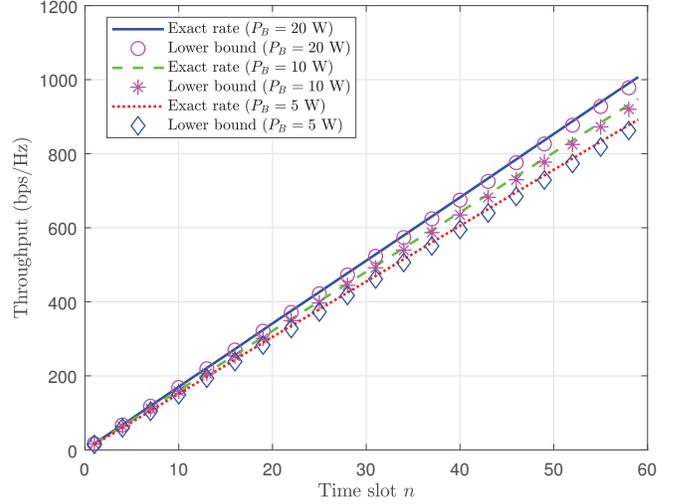}
	\caption{Exact receiving rate $ R_{U_r} $ and its lower bound $ \underline{R}_{U_r} $ with $ N=60 $ and $ \sigma^2=-114 $~dBm.}
	\label{Fig:RrBound}
\end{figure}

To exemplify the tightness of this bound, we here plot in Fig. \ref{Fig:RrBound} the exact receiving rate $ R_{U_r} $ and its lower bound $ \underline{R}_{U_r} $ for comparison. To make the figure more visible, the Y-axis represents the cumulative sum rate over time slots. From this figure, we can see that the adopted lower bound is rather tight, which can imply that we can get a near-optimal solution by exploiting this bound.

Then we need to deal with $ R_{U_t}[n] $ in (20c) and obtain an upper bound. To simplify the notations, we define 
\begin{align}\label{eq:r[n]}
r[n]&\triangleq \log_2\left(1+\frac{P_U h_{uk}[n]}{\sigma^2}\right) \nonumber\\
&=\log_2\left(1+\frac{P_U \alpha_0}{\sigma^2(H^2+\|\bm{u}[n]-\bm{e}_k\|^2)}\right),
\end{align}
\vspace{0.15cm}
by replacing $ \bm{u}[n] $ and $ \bm{e}_k $ with their horizontal coordinates with respect to $ x $ and $ y $, $ r[n] $ can be rewritten as
\begin{eqnarray}\label{eq:upper_r}
&&\!\!\!\!\!\!\!\!\!\!\!\!\!\!\!\!\!\!\!\!\!\!\!\!r[n]=\log_2\left(1+\frac{P_U \alpha_0}{\sigma^2(H^2+\|\bm{u}[n]-\bm{e}_k\|^2)}\right) \nonumber\\
&&\!\!\!\!\!\!\!\!\!\!\!\!\!=\log_2 \!\left(\!1\!+\!\frac{P_U \alpha_0}{\sigma^2\left[H^2 \!+\!(x[n] \!-\!x_k)^2 \!+\!(y[n] \!-\!y_k)^2\right]}\!\right)\nonumber\\
&&\!\!\!\!\!\!\!\!\!\!\!\!\!\leq \frac{\log_2\left(1+\frac{P_U \alpha_0}{3\sigma^2 H^2}\right)}{3} +\frac{\log_2\left(1+\frac{P_U \alpha_0}{3\sigma^2(x[n]-x_k)^2}\right)}{3}\nonumber\\
&&\!\!\!\!\!\!\!\!+\frac{\log_2\left(1+\frac{P_U \alpha_0}{3\sigma^2(y[n]-y_k)^2}\right)}{3} \nonumber\\
&&\!\!\!\!\!\!\!\!\!\!\!\!\!\triangleq \overline{r}[n],
\end{eqnarray}
\vspace{0.1cm}
where we use the Jensen's Inequality for the convex function $ \log\left(1+\frac{A}{x}\right) $ for any $ A>0 $. We can verify that the second term of $ \overline{r}[n] $ in \eqref{eq:upper_r} is convex with respect to $ x[n] $ and the third term is convex with respect to $ y[n] $ by checking their second-order derivatives. Then from the definition of $ r[n] $ in \eqref{eq:r[n]}, we have
\begin{eqnarray}\label{eq:Rtub}
&&\!\!\!\!\!\!\!\!\!\!\!\!\!\!\!\!\!\!\!\!R_{U_t}[n]=\sum_{k=1}^{K} \rho_{k,n} r[n] \nonumber\\
&&\leq \sum_{k=1}^{K} \rho_{k,n} \overline{r}[n] \nonumber\\
&&\triangleq \overline{R}_{U_t}[n].
\end{eqnarray}
\vspace{0.1cm}
Thus, we obtain an upper bound of $ R_{U_t}[n] $, i.e., $ \overline{R}_{U_t}[n] $, and it is convex with respect to $ \bm{u}[n] $. Namely, $ \overline{R}_{U_t}[n] $ is a convex upper bound of the LHS of (20c).

\begin{figure}[t!]
	\setlength{\abovecaptionskip}{-5pt}
	\centering
	\includegraphics[width = 3.7in]{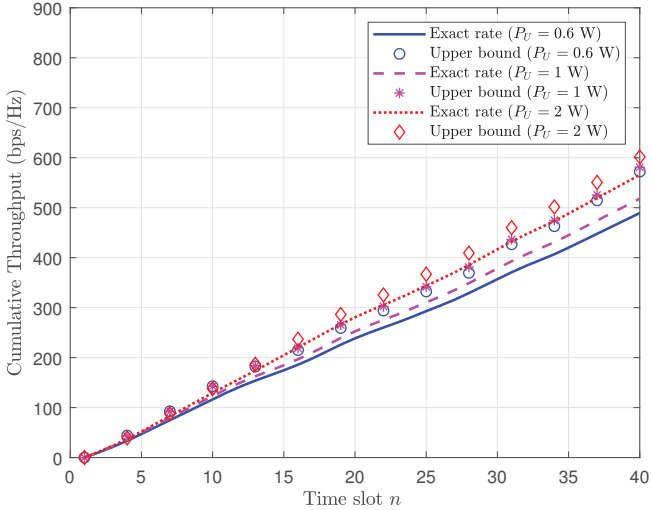}
	\caption{Comparison of $ R_{U_t} $ and its upper bound $ \overline{R}_{U_t} $.}
	\label{Fig:Rt_upper}
\end{figure}

Similarly, we compare the exact transmitting rate $ R_{U_t} $ and its upper bound $ \overline{R}_{U_t} $ in Fig. \ref{Fig:Rt_upper} to show the tightness of $ \overline{R}_{U_t} $. It is observed that the derived bound is fairly tight.

Then, as for the nonconvex objective function (20a) and constraint (20b), we need to determine a concave lower bound of $ R_E[k] $. $ r[n] $ defined in \eqref{eq:r[n]} is convex with respect to $ \|\bm{u}[n]-\bm{e}_k\|^2 $. Considering that the first-order Taylor approximation is a global under-estimator of convex functions, we have
\vspace{0.1cm}
\begin{eqnarray}\label{eq:lower_r}
&&\!\!\!\!\!\!\!\!\!\!\!\!r[n]\geq -c_k ^{(l)}[n](\|\bm{u}[n]\!-\!\bm{e}_k\|^2\!-\!\|\bm{u}^{(l)}[n]\!-\!\bm{e}_k\|^2)+d_k^{(l)}[n],\nonumber\\
&&\!\triangleq \underline{r}[n],
\end{eqnarray}
\vspace{0.1cm}
where
\begin{equation}\label{eq:c}
	c_k ^{(l)}[n]=\frac{\frac{P_U\alpha_0\sigma^2}{[\sigma^2(H^2+\|\bm{u}^{(l)}[n]-\bm{e}_k\|^2)]^2}\log_2 e}{1+\frac{P_U\alpha_0}{\sigma^2(H^2+\|\bm{u}^{(l)}[n]-\bm{e}_k\|^2)}}\geq 0,
\end{equation} \vspace{0.2cm}
\begin{equation}\label{eq:d}
	d_k ^{(l)}[n]=\log_2\left(1+\frac{P_U\alpha_0}{\sigma^2(H^2+\|\bm{u}^{(l)}[n]-\bm{e}_k\|^2)}\right).
\end{equation}
It is easy to check that $ \underline{r}[n] $ is concave with respect to $ \bm{u}[n] $ because the coefficient $ c_k ^{(l)}[n] $ is a nonnegative value. 

Further from \eqref{eq:rate_k} and \eqref{eq:lower_r}, the lower bound $ \underline{R}_E[k] $ is directly obtained as
\begin{equation}
\underline{R}_E[k]=\frac{1}{N} \sum_{n=1}^{N}\rho_{k,n} \underline{r}[n].
\end{equation}
\vspace{0.1cm}
which is concave with respect to $ \bm{u}[n] $. The comparison of $ R_E $ and its lower bound $ \underline{R}_E $ is plotted in Fig. \ref{Fig:Rk_lower}.

\begin{figure}[t!]
	\setlength{\abovecaptionskip}{-5pt}
	\centering
	\includegraphics[width = 3.6in]{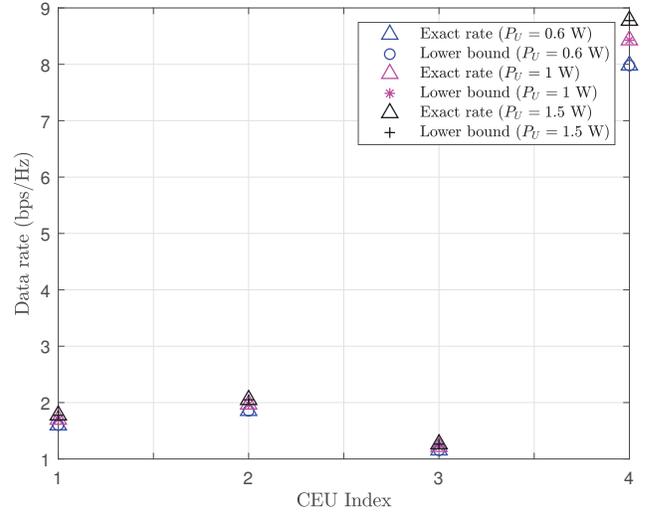}
	\caption{Comparison of $ R_E $ and its lower bound $ \underline{R}_E $.}
	\label{Fig:Rk_lower}
\end{figure}

By introducing the lower bounds $ \underline{R}_E[k] $ and $ \underline{R}_{U_r}[n] $ of $ R_E[k] $ and $ R_{U_r}[n] $ respectively and the upper bound, $ \overline{R}_{U_t}[n] $, of $ R_{U_t}[n] $, we successfully transform problem \eqref{eq:pro_U} into the following problem:
\vspace{0.1cm}
\begin{subequations}\label{eq:pro_U_new}
	\begin{align}
	\mathop{\max}_{\bm{U}, \bm{V}, \bm{A}}\quad\!\!
	&\sum_{k=1}^{K} \underline{R}_E[k]\\
	\textrm{s.t.}\quad\!\!\!\!
	&\underline{R}_E[k] \geq R_0, \ \forall k \in \mathcal{K},\\
	&\sum_{i=2}^{n}\overline{R}_{U_t}[i] \leq \sum_{i=1}^{n-1}\underline{R}_{U_r}[i], \ n=2,\cdots,N,\\
	&\bm{u}[1]=\bm{u}[N]=\bm{u}_0,\\
	&\|\bm{u}[n+1]\!-\!\bm{u}[n]\|\!\leq\!V_{\rm{max}} \delta_t, \ n\!=\!1,\cdots,N\!-\!1,\!\!\!\!\!\!\!\!\!\\
	&\bm{v}[n+1]=\bm{v}[n]+\bm{a}[n]\delta_t, \ n=1,\cdots,N-1,\\
	&\bm{u}[n\!\!+\!\!1]\!=\!\bm{u}[n]\!\!+\!\!\bm{v}[n]\delta_t\!\!+\!\!\frac{1}{2}\bm{a}[n]\delta_t^2, n\!=\!1,\!\cdots\!,N\!\!-\!\!1, \\
	&\|\bm{v}[n]\|\leq V_{\rm{max}}, \ n=1,\cdots,N, \\
	&\|\bm{a}[n]\|\leq a_{\rm{max}}, \ n=1,\cdots,N.
	\end{align}
\end{subequations}

\begin{mytheorem}
Problem \eqref{eq:pro_U_new} is a convex problem.
\end{mytheorem}

\begin{proof}
According to the above analysis, since $ \underline{R}_E[k] $ is concave with respect to $ \bm{u}[n] $, (31a) is a convex objective function and (31b) is a convex constraint. Similarly, since $ \overline{R}_{U_t}[n] $ is convex with respect to $ \bm{u}[n] $ and $ \underline{R}_{U_r}[n] $ is concave with respect to $ \bm{u}[n] $, thus (31c) is a convex constraint. Furthermore, (31d), (31f), and (31g) are linear constraints. (31e), (31h), and (31i) are convex constraints. Therefore, problem \eqref{eq:pro_U_new} is a convex problem.
\end{proof}
This convex problem can then be efficiently solved by using the well-established standard convex optimization method such as the interior-point method \cite{Boyd2004Convex}.

\subsection{UAV-CEU Association Optimization}\label{sec:association optimization}
Given a mobility management strategy of the UAV, the subproblem of optimizing the UAV-CEU association is rewritten from problem \eqref{eq:pro_U_P} as follows:
\begin{subequations}\label{eq:pro_P}
	\begin{align}
	\mathop{\max}_{\bm{P}}\quad\!\!
	&\sum_{k\in \mathcal{K}}R_E[k]\\
	\textrm{s.t.}\quad\!\!
	&R_E[k]\geq R_0, \ \forall k\in \mathcal{K}, \\
	&\sum_{k\in \mathcal{K}}\rho_{k,n}\leq 1, \ n=2,\cdots,N,\\
	&\rho_{k,n}=\{0,1\}, \ \forall k\in \mathcal{K}, \ n=2,\cdots,N,\\
	&\sum_{i=2}^{n}R_{U_t}[i]\leq \sum_{i=1}^{n-1}R_{U_r}[i], \ n=2,\cdots,N.
	\end{align}
\end{subequations}

It is difficult to solve problem \eqref{eq:pro_P} because of the integer variable $ \rho_{k,n} $. By relaxing (32d) to the continuous constraint $ \rho_{k,n}\in\left[0,1\right] $, problem \eqref{eq:pro_P} reduces to a standard linear programming because the objective function and the constraints are linear combinations of $ \bm{P} $. The linear programming is
\begin{subequations}\label{eq:pro_P_linear}
	\begin{align}
		\mathop{\max}_{\bm{P}}\quad\!\!
		&\sum_{k\in \mathcal{K}}R_E[k]\\
		\textrm{s.t.}\quad\!\!
		&0\leq \rho_{k,n}\leq 1,\ \forall k\in \mathcal{K},\ n=2,\cdots,N, \\
		&(\text{32b}),(\text{32c}),(\text{32e}).
	\end{align}
\end{subequations}

Naturally, this linear programming is a convex optimization problem. Typically, like in \cite{Cheng2018UAV}, the relaxed problem was solved by classical optimization methods, and then the solution of the relaxed problem was rounded to get the desired integer results. However, in this way, the optimality of the solution can not be guaranteed in theory and the feasibility of the solution may not hold due to the operation of rounding. This motivates us to adopt the Lagrangian dual decomposition method to obtain a low-complexity solution. 

In the following, we show that it fortunately returns integer solutions, which preserves both optimality and feasibility of the original problem if the variable relaxation is also deployed temporarily but using the dual decomposition approach.

After relaxing the binary constraints with respect to $ \bm{P} $, problem \eqref{eq:pro_P} becomes a standard linear program. By introducing dual variables $ \bm{\lambda}=\{\lambda_n\}_{n=2}^{N} $ and $ \bm{\eta}=\{\eta_k\}_{k=1}^{K} $, we can write the Lagrangian function of problem \eqref{eq:pro_P_linear} as
\begin{align}\label{eq:LagFun}
L(\bm{P},\bm{\lambda},\bm{\eta})= &\sum_{k=1}^{K} R_E[k]-\sum_{n=2}^{N} \lambda_n \left(\sum_{i=2}^{n} R_{U_t}[i]-\sum_{i=1}^{n-1} R_{U_r}[i]\right) \nonumber \\
&-\sum_{k=1}^{K}\eta_k(R_0-R_E[k]),
\end{align}
where the dual variables $ \bm{\lambda}=\{\lambda_n\}_{n=2}^{N} $ and $ \bm{\eta}=\{\eta_k\}_{k=1}^{K} $ are all nonnegative.

Equivalently, we solve its dual problem
\begin{subequations}\label{eq:Dual}
	\begin{align}
	\mathop{\min}_{\bm{\lambda}\geq \mathbf{0},\bm{\eta}\geq \mathbf{0}} 
	&\mathop{\max_{\bm{P}}} L(\bm{P},\bm{\lambda},\bm{\eta}) \\
	\textrm{s.t.}\quad\!\!
	&0\leq \rho_{k,n}\leq 1,\ \forall k\in \mathcal{K},\ n=2,\cdots,N, \\
	&\sum_{k\in \mathcal{K}}\rho_{k,n}\leq 1,\ n=2,\cdots,N.	
	\end{align}
\end{subequations}
By defining
\begin{equation}
	m_{k,n}\triangleq \log_2\left(1+\frac{P_U h_{uk}[n]}{\sigma^2}\right),
\end{equation}
the inner maximization in \eqref{eq:Dual} is rewritten as
\begin{subequations}\label{eq:opt_P}
	\begin{align}
	\mathop{\max_{\bm{P}}} \frac{1}{N} 
	&\sum_{k=1}^{K} \!\sum_{n=2}^{N} \!\rho_{k,n} m_{k,n}\!(1+\eta_k)\! -\sum_{k=1}^{K}\sum_{n=2}^{N}\left(\lambda_n\!\sum_{i=2}^{n}\!\rho_{k,i}m_{k,i}\right) \label{eq:inner max}\\
	\textrm{s.t.}\quad\!\!
	&0\leq \rho_{k,n}\leq 1,\ \forall k\in \mathcal{K},\ n=2,\cdots,N, \label{eq:rho1}\\
	&\sum_{k\in \mathcal{K}}\rho_{k,n}\leq 1,\ n=2,\cdots,N.\label{eq:rho2}
	\end{align}
\end{subequations}

\begin{algorithm}[t!]
	Initialize $ \bm{P} $ and let $ l=0 $\;
	\Repeat{convergence}
	{Given $ \bm{P}^{(l)} $, find the optimal $ \bm{U}^{(l+1)} $, $ \bm{V}^{(l+1)} $, and $ \bm{A}^{(l+1)} $ by solving problem \eqref{eq:pro_U_new}\;
		Given $ \bm{U}^{(l+1)} $, $ \bm{V}^{(l+1)} $, and $ \bm{A}^{(l+1)} $, find the optimal $ \bm{P}^{(l+1)} $ by solving problem \eqref{eq:pro_P}\;
		Update $ l=l+1 $\;
	}
	
	Return the UAV mobility management $ \bm{U}^*=U^{(l)} $, $ \bm{V}^*=V^{(l)} $, $ \bm{A}^*=A^{(l)} $, and the corresponding UAV-CEU association strategy $ \bm{P}^*=P^{(l)} $.
	\caption{IMMUA Algorithm for Problem \eqref{eq:pro_U_P}}
\end{algorithm}

We simplify the objective function in \eqref{eq:inner max} further and rewrite it as
\begin{equation}\label{eq:simple inner max}
	\mathop{\max_{\bm{P}}} \sum_{k=1}^{K} \sum_{n=2}^{N}A_{k,n}\rho_{k,n},
\end{equation}
which is a linear combination of $ \rho_{k,n} $ and the coefficient $ A_{k,n}=\left(\frac{m_{k,n}(1+\eta_k)}{N}-\sum_{i=n}^{N}\lambda_i m_{k,n}\right) $. To obtain the maximum value of \eqref{eq:simple inner max}, we should let $ \rho_{k,n} $ with the largest coefficient be 1 and the others be 0 for any $ n $ due to the constraints \eqref{eq:rho1} and \eqref{eq:rho2}. Thus, it implies the optimal solution as
\begin{equation}\label{eq:optimal association}
\rho_{k,n}^*=
\begin{cases}
1,&\text{if}\; k \!=\! k^{(n)}\\
0,&\text{if}\; k \!\neq\! k^{(n)},
\end{cases}
\end{equation}
where $ k^{(n)}=\arg\;\max\limits_{q\in\mathcal{K}}\left(\frac{m_{q,n}(1+\eta_q)}{N}-\sum_{i=n}^{N}\lambda_i m_{q,n}\right) $.

Notice that the optimal $ \rho_{k,n}^* $ is proven in \eqref{eq:optimal association} to be either 0 or 1 which satisfies the integer constraint (32d) in problem \eqref{eq:pro_P}, even though we temporarily relaxed $ \rho_{k,n} $ as a continuous variable. Therefore, the optimal solution to problem \eqref{eq:opt_P} is exactly given by \eqref{eq:optimal association}.

Then, we need to solve the outer minimization in \eqref{eq:Dual} using the integer solution of $ \rho_{k,n}^* $ in \eqref{eq:optimal association}. In each iteration, we exploit the subgradient based method \cite{Boyd2004Convex} to update the dual variables as
\begin{equation}
\lambda_n^{(t+\!1\!)}\!=\!\left[\lambda_n^{(t)}\!-\!\delta^{(t)}\! \!\left(\!\!-\!\!\sum_{i=2}^{n}\sum_{k=1}^{K}\rho_{k,i}^{(t)}m_{k,i}^{(t)}\!+\!\sum_{i=1}^{n-1}R_{U_r}^{(t)}[i] \!\right)\!\right]^+,
\end{equation}
\begin{equation}
\eta_k^{(t+1)}=\left[\eta_k^{(t)}-\delta^{(t)}\left(-R_0+\frac{1}{N}\sum_{n=2}^{N}\rho_{k,n}^{(t)} m_{k,n}^{(t)}\right)\right]^+,
\end{equation}
where $ \delta^{(t)} $ is the step size and 
\begin{equation}
	[x]^+=
	\begin{cases}
	x,&\text{if}\;x\geq0\\
	0,&\text{if}\;x<0.
	\end{cases}
\end{equation}
We update dual variables $ \bm{\lambda}=\{\lambda_n\}_{n=2}^{N} $ and $ \bm{\eta}=\{\eta_k\}_{k=1}^{K} $ and association indicators $ \bm{P}=\{\rho_{k,n},\forall k\in \mathcal{K},n=2,\cdots,N\} $ iteratively until the objective function in \eqref{eq:Dual} converges. In this way, the UAV-CEU association optimization problem in \eqref{eq:pro_P_linear} is solved.

To this end, we are able to solve the original problem by tackling the two subproblems, i.e., UAV mobility management and UAV-CEU association optimization, in an alternating manner. We summarize the iterative mobility management and user association (IMMUA) algorithm in \textbf{Algorithm 1}, which can obtain a suboptimal solution with low complexity. The convergence and complexity of IMMUA algorithm is analyzed in the following subsection. Even though we first decompose the original problem into two subproblems and then solve them separately in two steps, we obtain optimal or near-optimal solution in both steps. This guarantees good performance of our proposed IMMUA algorithm which will be demonstrated by numerical results in Section \ref{sec:results}.

\begin{mydis}
\textup{In addition to the objective function formulated in problem \eqref{eq:pro_U_P}, our proposed IMMUA algorithm can also handle some other forms of objective functions if the traffic patterns are considered. For instance, we can maximize the weighted sum rate of CEUs which is expressed as}
\begin{equation} \label{eq:function_new}
\mathop{\max}_{\bm{U},\bm{V},\bm{A},\bm{P}} \ \sum_{k\in \mathcal{K}}w_k R_E[k],
\end{equation}
\textup{where $ w_k $ denotes the constant weight of CEU $ k $. Note that the values of $ w_k $ can be determined by the traffic patterns, e.g., Poisson distribution, for different users \cite{Kobayashi2006An}. Since the weights are constant, they do not affect the application of our proposed algorithm to solve the problem.}

\textup{Considering more directly the traffic arrival patterns for different users, we can add minimum rate requirement of each user individually. That is to change (19b) in problem \eqref{eq:pro_U_P} into the following constraint:}
\begin{equation}
R_E[k]\geq R_k, \ \forall k \in \mathcal{K},
\end{equation}
\textup{where $ R_k $ is the minimum rate requirement of user $ k $. In particular, $ R_k $ can be determined according to user types, mission types, and traffic arrival patterns. In this way, the new problem imposes different QoS requirements on different users. This new constraint can be handled by using the similar bounding techniques for (19b).} 
\end{mydis}

\subsection{Convergence and Complexity Analysis}\label{sec:complexity}
With our proposed IMMUA algorithm, the resulting objective function value of problem \eqref{eq:pro_U_P} is non-decreasing after each iteration. Furthermore, it has a finite upper bound. Therefore, the overall IMMUA algorithm is guaranteed to converge.

The complexity of the IMMUA algorithm lies in solving the UAV mobility management problem and the UAV-CEU association optimization problem. Considering that we solve the UAV mobility management problem via the standard interior-point method and the number of optimization variables is $ 6N $, the complexity of solving this problem is $ \mathcal{O}(L_i N^3) $ [10, Pages 487, 569], where $ L_i $ denotes the number of iterations required by the interior-point method. As for the UAV-CEU association optimization problem, we solve it via the dual decomposition method, whose complexity is $ \mathcal{O}(L_d (N+K)) $, where $ L_d $ denotes the number of iterations needed by the dual method. Therefore, the total complexity of ITUA algorithm is $ \mathcal{O}(L_o(L_i N^3+L_d (N+K))) $, where $ L_o $ represents the number of outer iterations.

\section{Numerical Results}\label{sec:results}
In this section, numerical results are presented to validate the effectiveness of our proposed algorithm. We consider a UAV-assisted communication network in IoT applications. As depicted in Fig. \ref{Fig:trajectory}, the network consists of one UAV, three adjacent GBSs, and four CEUs. Specifically, these CEUs are any type of IoT devices, which are randomly deployed in the overlapped coverage of GBSs to perform specific tasks. Due to the long distance from GBSs, the channel quality between CEUs and GBSs is poor and the QoS of CEUs cannot be guaranteed. Under this circumstance, a UAV acts as a mobile relay to forward data from GBSs to CEUs when data transmission is needed, which is a convenient and cost-efficient way to improve the QoS of CEUs without increasing infrastructure construction. We need to note that the above simulation setup also applies to a post-disaster communication scenario where a sudden disaster damaged the cellular infrastructure in the area affected by the disaster \cite{Zhao2019UAV}. By leveraging a UAV, rescue information can be sent from the GBSs in the area unaffected by the disaster to the IoT devices in the disaster area. Each GBS is equipped with $ L=8 $ antennas. The radius of each cell is set to $ r=1000 $~m and the horizontal locations of three GBSs are $ (0,r) $, $ (\sqrt{3}r,r) $, and $ (\frac{\sqrt{3}}{2}r,-\frac{1}{2}r) $ respectively. Other important simulation parameters are listed in Table \ref{TABLE:parameter} unless otherwise specified.

\begin{table}[t!]
	\small
	\setlength{\belowcaptionskip}{-12pt}
	\centering
	\caption{Simulation Parameters}\label{TABLE:parameter}
	\begin{tabular}{|c|c|c|}
		\hline
		\textbf{Parameter} & \textbf{Description} & \textbf{Value} \\
		\hline
		$ a_{\rm{max}} $ & Maximum UAV acceleration & $ 5~\rm{m/s^2} $ \\
		\hline
		$ \alpha_0 $ & Reference channel power at $ d_0=1 $~m & -60 dB \\ 
		\hline
		$ H $ & Flight altitude of UAV & 100 m  \\
		\hline
		$ L $ & Number of antennas at GBSs & 8 \\
		\hline
		$ P_B $ & Transmit power of GBSs & 10 W \\
		\hline
		$ P_U $ & Transmit power of UAV  & 1 W \\
		\hline
		$ r $ & Radius of the cell & 1000 m\\ 
		\hline
		$ r_c $ & Radius of the circular trajectory & 500 m \\
		\hline
		$ R_0 $ & Minimum rate requirement of CEU & 0.5 bps/Hz \\
		\hline
		$ V_{\rm{max}} $ & Maximum UAV speed & 50 m/s \\
		\hline
		$ \sigma^2 $ & Noise power & -114 dBm \\
		\hline	
	\end{tabular}
\end{table}

\begin{table}[t!]
	\scriptsize
	\setlength{\belowcaptionskip}{-10pt}
	\centering
	\caption{UAV-CEU association strategy obtained by IMMUA algorithm with the maximum speed of $ V_{\rm{max}} = 40 $~m/s}\label{TABLE:associationV40}
	\begin{tabular}{|c|c|c|c|c|c|c|c|c|c|c|c|c|c|c|c|}
		\hline
		\textbf{Time slot} & \textbf{1} & \textbf{2} & \textbf{3} & \textbf{4} & \textbf{5} & \textbf{6} & \textbf{7} & \textbf{8} & \textbf{9} & \textbf{10} \\
		\hline
		\textbf{Associated CEU} & Null & 2 & 4 & 4 & 4 & 4 & 4 & 4 & 4 & 4 \\ 
		\hline
		\textbf{Time slot} & \textbf{11} & \textbf{12} & \textbf{13} & \textbf{14} & \textbf{15} & \textbf{16} & \textbf{17} & \textbf{18} & \textbf{19} & \textbf{20}   \\
		\hline
		\textbf{Associated CEU} & 4 & 4 & 4 & 4 & 1 & 1 & 1 & 1 & 1 & 1 \\
		\hline
		\textbf{Time slot} & \textbf{21} & \textbf{22} & \textbf{23} & \textbf{24} & \textbf{25} & \textbf{26} & \textbf{27} & \textbf{28} & \textbf{29} & \textbf{30} \\
		\hline
		\textbf{Associated CEU} & 1 & 1 & 4 & 4 & 2 & 2 & 2 & 2 & 2 & 2 \\
		\hline
		\textbf{Time slot} & \textbf{31} & \textbf{32} & \textbf{33} & \textbf{34} & \textbf{35} & \textbf{36} & \textbf{37} & \textbf{38} & \textbf{39} & \textbf{40} \\
		\hline
		\textbf{Associated CEU} & 2 & 2 & 4 & 4 & 4 & 4 & 3 & 3 & 3 & 3  \\
		\hline
		\textbf{Time slot} & \textbf{41} & \textbf{42} & \textbf{43} & \textbf{44} & \textbf{45} & \textbf{46} & \textbf{47} & \textbf{48} & \textbf{49} & \textbf{50} \\ 
		\hline
		\textbf{Associated CEU} & 4 & 4 & 4 & 4 & 4 & 4 & 4 & 4 & 4 & 3 \\
		\hline	
		\textbf{Time slot} & \textbf{51} & \textbf{52} & \textbf{53} & \textbf{54} & \textbf{55} & \textbf{56} & \textbf{57} & \textbf{58} & \textbf{59} & \textbf{60} \\
		\hline
		\textbf{Associated CEU} & 4 & 3 & 3 & 3 & 4 & 4 & 4 & 4 & 4 & 2\\
		\hline
	\end{tabular}
\end{table}

\begin{table}[t!]
	\scriptsize
	\setlength{\belowcaptionskip}{-10pt}
	\centering
	\caption{UAV-CEU association strategy obtained by IMMUA algorithm with the maximum speed of $ V_{\rm{max}} = 50 $~m/s}\label{TABLE:associationV50}
	\begin{tabular}{|c|c|c|c|c|c|c|c|c|c|c|c|c|c|c|c|}
		\hline
		\textbf{Time slot} & \textbf{1} & \textbf{2} & \textbf{3} & \textbf{4} & \textbf{5} & \textbf{6} & \textbf{7} & \textbf{8} & \textbf{9} & \textbf{10} \\
		\hline
		\textbf{Associated CEU} & Null & 2 & 4 & 4 & 4 & 4 & 4 & 4 & 4 & 4 \\ 
		\hline
		\textbf{Time slot} & \textbf{11} & \textbf{12} & \textbf{13} & \textbf{14} & \textbf{15} & \textbf{16} & \textbf{17} & \textbf{18} & \textbf{19} & \textbf{20}   \\
		\hline
		\textbf{Associated CEU} & 4 & 4 & 4 & 4 & 1 & 1 & 1 & 1 & 1 & 1 \\
		\hline
		\textbf{Time slot} & \textbf{21} & \textbf{22} & \textbf{23} & \textbf{24} & \textbf{25} & \textbf{26} & \textbf{27} & \textbf{28} & \textbf{29} & \textbf{30} \\
		\hline
		\textbf{Associated CEU} & 1 & 1 & 4 & 4 & 4 & 4 & 2 & 2 & 2 & 2 \\
		\hline
		\textbf{Time slot} & \textbf{31} & \textbf{32} & \textbf{33} & \textbf{34} & \textbf{35} & \textbf{36} & \textbf{37} & \textbf{38} & \textbf{39} & \textbf{40} \\
		\hline
		\textbf{Associated CEU} & 2 & 2 & 2 & 4 & 4 & 4 & 3 & 3 & 3 & 3  \\
		\hline
		\textbf{Time slot} & \textbf{41} & \textbf{42} & \textbf{43} & \textbf{44} & \textbf{45} & \textbf{46} & \textbf{47} & \textbf{48} & \textbf{49} & \textbf{50} \\ 
		\hline
		\textbf{Associated CEU} & 4 & 4 & 4 & 4 & 4 & 4 & 4 & 4 & 4 & 4 \\
		\hline	
		\textbf{Time slot} & \textbf{51} & \textbf{52} & \textbf{53} & \textbf{54} & \textbf{55} & \textbf{56} & \textbf{57} & \textbf{58} & \textbf{59} & \textbf{60} \\
		\hline
		\textbf{Associated CEU} & 4 & 4 & 4 & 4 & 3 & 4 & 4 & 4 & 4 & 2\\
		\hline
	\end{tabular}
\end{table}

\begin{figure}[htbp]
	\centering
	\subfigure[the optimized trajectory with $ V_{\rm{max}}=40 $~m/s]{
	\includegraphics[width = 3.6in]{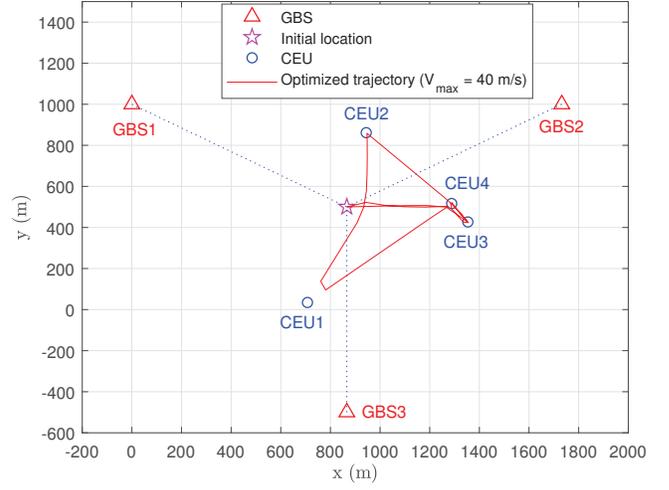} \label{Fig:V40}
	}
    \subfigure[the optimized trajectory with $ V_{\rm{max}}=50 $~m/s]{
    \includegraphics[width = 3.6in]{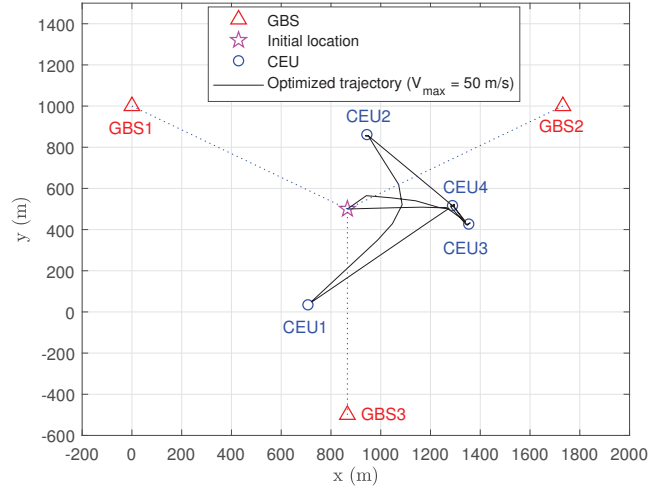} \label{Fig:V50}
    }
\caption{Optimized UAV trajectories obtained by IMMUA algorithm with the maximum speed of $ V_{\rm{max}}=40 $~m/s and $ V_{\rm{max}}=50 $~m/s, respectively.}
\label{Fig:trajectory}
\end{figure}

\begin{figure}[t!]
	\setlength{\abovecaptionskip}{-5pt}
	\centering
	\includegraphics[width = 3.6in]{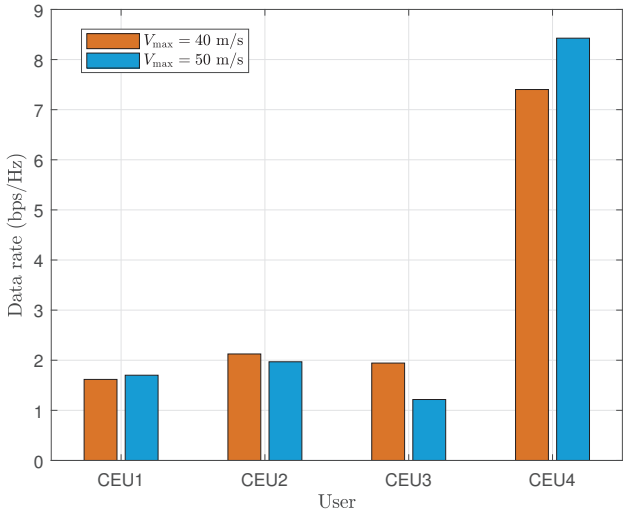}
	\caption{The rate of each CEU with the optimized mobility management and association strategy.}
	\label{Fig:bar}
\end{figure}


Firstly, we show the UAV trajectory and UAV-CEU association strategy obtained by IMMUA algorithm intuitively. Fig. \ref{Fig:V40} and Fig. \ref{Fig:V50} illustrate the optimized UAV trajectories with the maximum UAV speeds of 40 m/s and 50 m/s, respectively. The pre-specified initial location of the UAV is set to the intersection point of the three cells, i.e., $ (\frac{\sqrt{3}}{2}r,\frac{1}{2}r) $. The UAV has to return to the initial location after each transmission period. As shown in Fig. \ref{Fig:trajectory}, the UAV flies in the following way: firstly, the UAV flies from the initial location to CEU4. Secondly, it flies to CEU1, and then flies to CEU2. After that, it flies through CEU4 to CEU3 and stays there for a while. Finally, it returns to the initial location.

For the cases of $ V_{\rm{max}}=40 $~m/s and $ V_{\rm{max}}=50 $~m/s, the corresponding UAV-CEU association strategies are presented in Table \ref{TABLE:associationV40} and Table \ref{TABLE:associationV50}, respectively. In the first time slot, no CEU is associated with the UAV due to the information causality constraint. According to the association strategy, we find that the UAV in fact associates with the nearest CEU in each time slot during its flight to maximize the sum rate of all CEUs. Based on the UAV mobility management strategy and UAV-CEU association strategy obtained by IMMUA algorithm, the data rates of the four CEUs are calculated and presented in Fig. \ref{Fig:bar}, respectively. CEU4 has the highest data rate since it has the most time slots associated to the UAV.
 

Fig. \ref{Fig:speed} illustrates the effect of acceleration constraint on UAV's speed verus the flying time slot. We can see that when the acceleration constraint is not taken into account, the speed changes very quickly, even close to infinity, which is not practical. On the contrary, when the acceleration constraint is considered and the maximum acceleration is set to $ a_{\rm{max}}=5~\rm{m/s^2} $, the change of speed is much gentler. The acceleration of the UAV over time slot is plotted in Fig. \ref{Fig:acceleration}, which shows that the acceleration cannot exceed $ a_{\rm{max}}=5~\rm{m/s^2} $.

\begin{figure}[t!]
	\setlength{\abovecaptionskip}{-5pt}
	\centering
	\includegraphics[width = 3.6in]{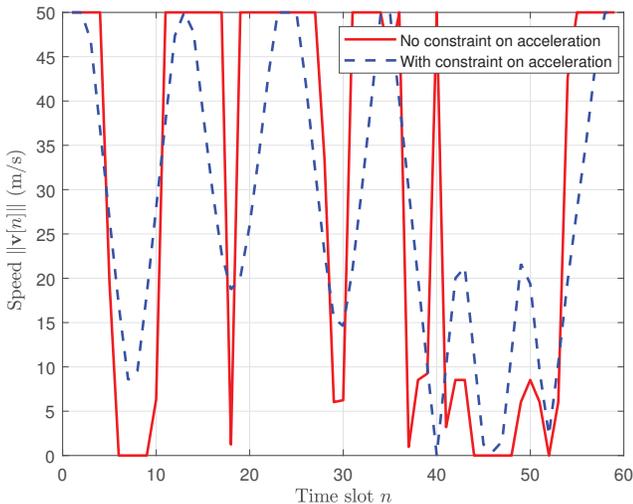}
	\caption{The speed of the UAV over time with $ V_{\rm{max}}=50 $~m/s.}
	\label{Fig:speed}
\end{figure}

\begin{figure}[t!]
	\setlength{\abovecaptionskip}{-5pt}
	\centering
	\includegraphics[width = 3.6in]{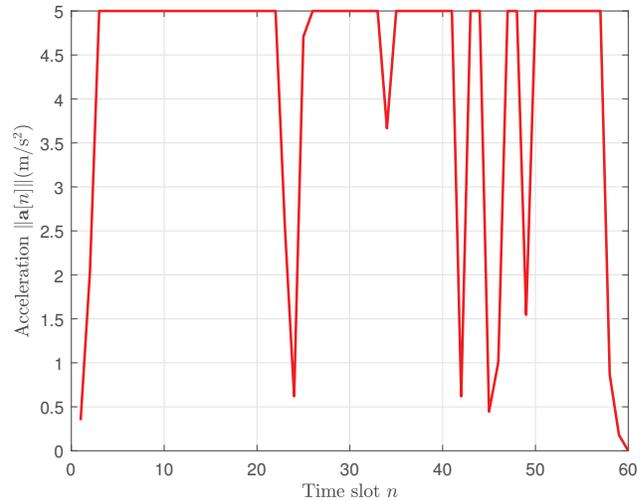}
	\caption{The acceleration of the UAV over time with $ a_{\rm{max}}=5~\rm{m/s^2} $.}
	\label{Fig:acceleration}
\end{figure}

\begin{figure}[t!]
	\setlength{\abovecaptionskip}{-5pt}
	\centering
	\includegraphics[width = 3.6in]{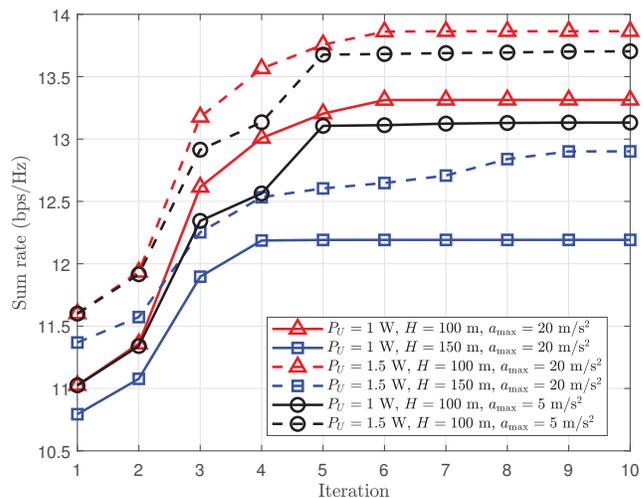}
	\caption{Sum rate of all CEUs versus the number of iterations with different $ P_U $, $ H $, and $ a_{\rm{max}} $.}
	\label{Fig:Throughput}
\end{figure}

Then, we verify the convergence behaviour of our proposed IMMUA algorithm. In Fig. \ref{Fig:Throughput}, we show the sum rate of all the CEUs versus the number of iterations with different values of UAV's transmit power $ P_U $, flight altitude $ H $, and maximum acceleration $ a_{\rm{max}} $. It indicates that our proposed algorithm converges in a few iterations as expected. From the figure results, we can conclude that, for most cases, 7 to 8 iterations can be sufficient for the algorithm to converge.
While in practice use, 5 iterations in most cases achieve a 99\% of the rate upon convergence. From the figure, the sum rate improves as the UAV's transmit power increases. This is intuitively true since an increase of transmit power could lead to a higher SINR, which implies higher data rate. On the other hand, the sum rate decreases as the UAV flies higher. This is because higher flight induces weak channel power gains resulting in lower rates. In theory, the lower UAV flies, the higher rate of CEUs. However, this is based on the assumption that the flight altitude of UAV is at least 100 m. When the flight altitude of UAV is much lower, the air-to-ground communication links will be blocked and scattered by buildings and other obstacles. The channel model used in this paper will no longer be applicable and the data rate will decrease. Therefore, the flight altitude of UAV cannot be too high or too low. Moreover, as can be seen from Fig. \ref{Fig:Throughput}, a smaller maximum acceleration leads to a lower sum rate, which is intuitively true since the feasible set is reduced.

\begin{figure}[t!]
	\setlength{\abovecaptionskip}{-5pt}
	\centering
	\includegraphics[width = 3.7in]{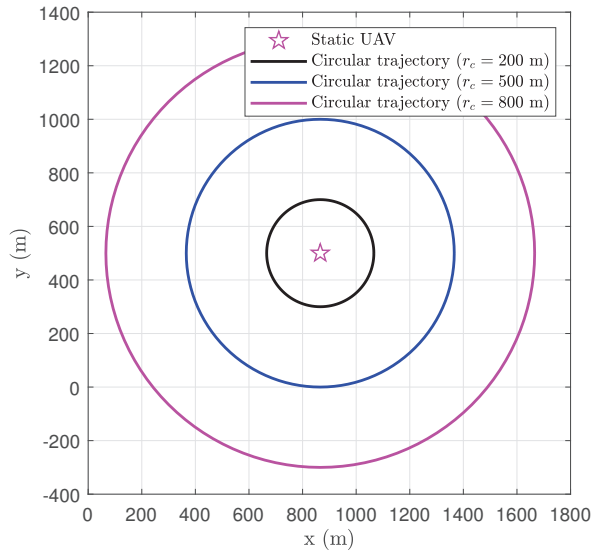}
	\caption{Benchmark trajectories: (a) static UAV, (b) circular trajectories with different radius.}
	\label{Fig:Circular}
\end{figure}

\begin{figure}[t!]
	\setlength{\abovecaptionskip}{-5pt}
	\centering
	\includegraphics[width = 3.5in]{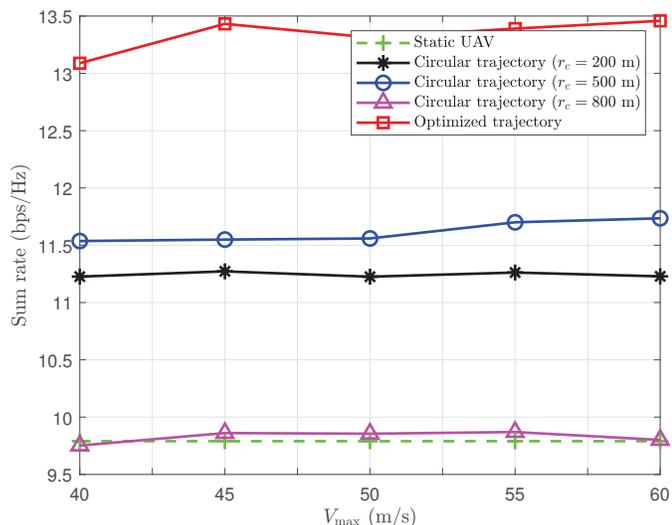}
	\caption{Sum rate of CEUs with different trajectories: static UAV, circle trajectories and optimized trajectory.}
	\label{Fig:rateVSspeed}
\end{figure}

\begin{figure}[t!]	
	\setlength{\abovecaptionskip}{-5pt}
	\centering
	\includegraphics[width = 3.7in]{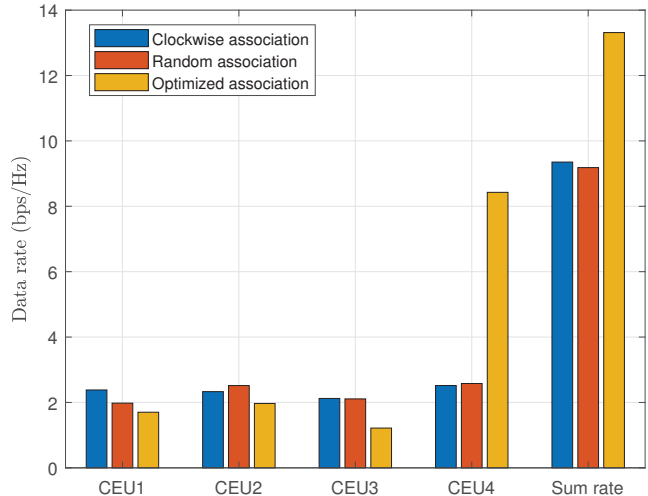}
	\caption{The data rate of each CEU with different UAV-CEU association strategies.}
	\label{Fig:rateVSassociation}
\end{figure}

\begin{figure}[t!]	
	\setlength{\abovecaptionskip}{-5pt}
	\centering
	\includegraphics[width = 3.6in]{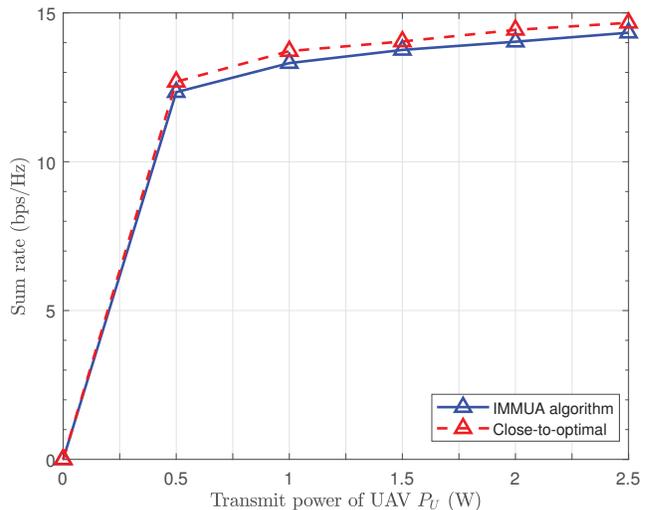}
	\caption{Performance comparison of IMMUA algorithm and close-to-optimal algorithm.}
	\label{Fig:optimal}
\end{figure}

For the purpose of comparison, we consider two types of benchmark UAV trajectories as shown in Fig. \ref{Fig:Circular}: (a) static UAV, where the UAV stays at the intersection point in the whole transmission period; (b) circle trajectories, where the UAV flies at the maximum speed with different circle radii, 200 m, 500 m, and 800 m, respectively. All the circle trajectories are centered at the intersection point. For different UAV trajectories, the UAV-CEU association strategy is optimized using our proposed algorithm as stated in Section \ref{sec:association optimization}. Fig. \ref{Fig:rateVSspeed} compares the sum rate of our proposed algorithm with those of the benchmark trajectories. It is observed that our proposed algorithm outperforms the benchmark schemes at different values of $ V_{\rm{max}} $, which demonstrates the efficiency of our proposed algorithm. The static UAV corresponds to a traditional relay, which has no degree of freedom. Therefore, the sum rate of CEUs does not vary with $ V_{\rm{max}} $. As for circle trajectories with different radii, the performance is quite different in terms of sum rate of CEUs. Specifically, the circle trajectory with $ r_c=500 $~m has the highest sum rate among the three circle trajectories due to the fact that this trajectory is close to the locations of CEUs. On the contrary, the circle trajectory with $ r_c=800 $~m performs worst since it is too far from CEUs. The performance of the circle trajectory with $ r_c=200 $~m is between the two. Besides, we observe that our proposed algorithm achieve 18.6\%, 13.4\% and 35.8\% sum rate gain over circle trajectories with $ r_c=200 $~m, 500 m and 800 m, respectively.

Then we consider the scenario where the UAV-CEU association strategy is fixed. The UAV trajectory is optimized using our proposed algorithm as stated in Section \ref{sec:trajectory optimization}. We consider random association strategy and clockwise association strategy as benchmarks. As for the random association strategy, the UAV associates with one CEU randomly in each time slot. According to the locations of CEUs, clockwise association strategy means the UAV associates with CEU1, CEU2, CEU4 and CEU3 in a clockwise manner. Fig. \ref{Fig:rateVSassociation} shows the data rate of each CEU with different association strategies. It shows that the optimized association strategy improves sum rate performance over benchmark association strategies at the cost of fairness.

As stated in Section \ref{sec:algorithm}, the original problem is a mixed integer program and the causality constraint is particularly hard to solve. It is challenging to obtain an optimal solution of such problem due to extremely high complexity. Even for the subproblem of mobility management, it is still nonconvex and tough to solve. Therefore, we compare the performance of our proposed algorithm with a close-to-optimal method that randomly selects 100 initial points for the proposed algorithm and returns the solution corresponding to the maximum objective value. Note that this method has been popular in approaching the optimal solution in literature \cite{Yang2018Joint,Chen2017Caching}. From Fig. \ref{Fig:optimal}, it can be seen that the performance of our proposed IMMUA algorithm quite approaches the close-to-optimal performance, which indicates the effectiveness of our proposed algorithm. Under different parameter settings, the proposed algorithm only causes less than 3\% performance loss.

\section{Conclusion}\label{sec:conclusion}
In this paper, we have studied a new mobile relaying technique with a cache-enabled UAV in a multi-cell network. We jointly optimized the UAV-CEU association strategy and UAV mobility management strategy to maximize the sum rate of all CEUs, subject to the minimum rate requirements of CEUs, mobility constraints and casual buffer constraints. We formulate an optimization problem and the original mixed-integer nonconvex problem is successfully transformed into two convex subproblems. Accordingly, an efficient iterative algorithm was developed to solve the two subproblems in an alternating manner, which is guaranteed to converge with low complexity. According to simulation results, the mobility of UAV induces rate improvement compared with static relay. Furthermore, our proposed algorithm performs well and outperforms the traditional trajectories and association strategies significantly. In our future work, we will extend the results obtained in this paper by taking into account the optimization of UAV's altitude, mobile CEUs, more UAVs as well as the energy efficiency of UAV.


\begin{thebibliography}{}
	\footnotesize
	
	\bibitem{Valavanis2015Handbook}
	K. P. Valavanis and G. J. Vachtsevanos, \emph{Handbook of Unmanned Aerial Vehicles.} Dordrecht, The Netherlands: Springer, 2015.
	
	\bibitem{Zeng2016Wireless}
	Y. Zeng, R. Zhang, and T. J. Lim, ``Wireless communications with unmanned aerial vehicles: Opportunities and challenges,'' \emph{IEEE Commun. Mag.}, vol. 54, no. 5, pp. 36--42, May 2016.
	
	\bibitem{Chen2018Liquid}
	M. Chen, W. Saad, and C. Yin, ``Liquid state machine learning for resource and cache management in LTE-U unmanned aerial vehicle networks,'' in \emph{Proc. IEEE GLOBECOM}, Singapore, Dec. 2017, pp. 1--6.
	
	\bibitem{Zhao2018Integrating}
	J. Zhao, F. Gao, G. Ding, T. Zhang, W. Jia, and A. Nallanathan, ``Integrating communications and control for UAV systems: Opportunities and challenges,'' \emph{IEEE Access}, vol. 6, pp. 67519--67527, 2018.
	
	\bibitem{Mozaffari2019Beyond}
	M. Mozaffari, A. T. Z. Kasgari, W. Saad, M. Bennis, and M. Debbah, ``Beyond 5G with UAVs: Foundations of 3D wireless cellular network,'' \emph{IEEE Trans. Wireless Commun.}, vol. 18, no. 1, pp. 357--372, Jan. 2019.
	
	\bibitem{Li2019UAV}
	B. Li, Z. Fei, and Y. Zhang, ``UAV communications for 5G and beyond: Recent advances and future trends,'' \emph{IEEE Internet Things J.}, vol. 6, no. 2, pp. 2241--2263, Apr. 2019.
	
	\bibitem{Hourani2014Optimal}
	A. Al-Hourani, S. Kandeepan, and S. Lardner, ``Optimal LAP altitude for maximum coverage,'' \emph{IEEE Wireless Commun. Lett.}, vol. 3, no. 6, pp. 569--572, Dec. 2014.
	
	\bibitem{Lyu2017Placement}
	J. Lyu, Y. Zeng, R. Zhang, and T. J. Lim, ``Placement optimization of UAV-mounted mobile base stations,'' \emph{IEEE Commun. Lett.}, vol. 21, no. 3, pp. 604--607, Mar. 2017.
	
	\bibitem{Wang2018Power}
	H. Wang, G. Ding, F. Gao, J. Chen, J. Wang, and L. Wang, ``Power control in UAV-supported ultra dense networks: Communications, caching, and energy transfer,'' \emph{IEEE Commun. Mag.}, vol. 56, no. 6, pp. 28--34, Jun. 2018.
	
	\bibitem{Duan2019Resource}
	R. Duan, J. Wang, C. Jiang, H. Yao, Y. Ren and Y. Qian, ``Resource allocation for multi-UAV aided IoT NOMA uplink transmission systems,'' \emph{IEEE Internet Things J.}, to be published.
	
	\bibitem{Zhan2011Wireless}
	P. Zhan, K. Yu, and A. L. Swindlehurst, ``Wireless relay communications with unmanned aerial vehicles: Performance and optimization,'' \emph{IEEE Trans. Aerosp. Electron. Syst.}, vol. 47, no. 3, pp. 2068--2085, Jul. 2011.
	
	\bibitem{Zeng2016Throughput}
	Y. Zeng, R. Zhang, and J. L. Teng, ``Throughput maximization for UAV-enabled mobile relaying systems,'' \emph{IEEE Trans. Commun.}, vol. 64, no. 12, pp. 4983--4996, Dec. 2016.
	
	\bibitem{Zhang2016Probabilistic}
	R. Zhang, M. Wang, X. Shen, and L. L. Xie, ``Probabilistic analysis on QoS provisioning for Internet of Things in LTE-A heterogeneous networks with partial spectrum usage,'' \emph{IEEE Internet Things J.}, vol. 3, no. 3, pp. 354--365, Jun. 2016.
	
	\bibitem{Motlagh2016Low}
	N. H. Motlagh, T. Taleb, and O. Arouk, ``Low-altitude unmanned aerial vehicle-based Internet of Things services: Comprehensive survey and future perspectives, '' \emph{IEEE Internet Things J.}, vol. 3, no. 6, pp. 899--922, Dec. 2016.
	
	\bibitem{Feng2019UAV}
	W. Feng, J. Wang, Y. Chen, X. Wang, N. Ge, and J. Lu, ``UAV-aided MIMO communications for 5G Internet of Things,'' \emph{IEEE Internet Things J.}, vol. 6, no. 2, pp. 1731--1740, Apr. 2019.
	
	\bibitem{Mozaffari2016Unmanned}
	M. Mozaffari, W. Saad, M. Bennis, and M. Debbah, ``Unmanned aerial vehicle with underlaid device-to-device communications: Performance and tradeoffs,'' \emph{IEEE Trans. Wireless Commun.}, vol. 15, no. 6, pp. 3949--3963, Jun. 2016.
	
	\bibitem{Su2020UAV}
	C. Su, F. Ye, L-C. Wang, L. Wang, Y. Tian, and Z. Han, ``UAV-assisted wireless charging for energy-constrained IoT devices using dynamic matching,'' \emph{IEEE Internet Things J.}, to be published, 2020.
	
	\bibitem{Liu2019Resource}
	X. Liu and N. Ansari, ``Resource allocation in UAV-assisted M2M communications for disaster rescue,'' \emph{IEEE Wireless Commun. Lett.}, vol. 8, no. 2, pp. 580--583, Apr. 2019.
	
	\bibitem{He2018Joint}
	H. He, S. Zhang, Y. Zeng, and R. Zhang, ``Joint altitude and beamwidth optimization for UAV-enabled multiuser communications,'' \emph{IEEE Commun. Lett.}, vol. 22, no. 2, pp. 344--347, Feb. 2018.
	
	\bibitem{Pan2019Joint}
	C. Pan, H. Ren, Y. Deng, M. Elkashlan, and A. Nallanathan, ``Joint blocklength and location optimization for URLLC-enabled UAV relay systems,'' \emph{IEEE Wireless Commun. Lett.}, vol. 23, no. 3, pp. 498--501, Mar. 2019.
	
	\bibitem{Zhang2018Trajectory}
	G. Zhang, H. Yan, Y. Zeng, M. Cui, and Y. Liu, ``Trajectory optimization and power allocation for multi-hop UAV relaying communications,'' \emph{IEEE Access}, vol. 6, pp. 48566--48576, 2018.
	
	\bibitem{Zhang2018Joint}
	S. Zhang, H. Zhang, Q. He, K. Bian, and L. Song, ``Joint trajectory and power optimization for UAV relay networks,'' \emph{IEEE Commun. Lett.}, vol. 22, no. 1, pp. 161--164, Jan. 2018.

	\bibitem{Mozaffari2016Efficient}
    M. Mozaffari, W. Saad, M. Bennis, and M. Debbah, ``Efficient deployment of multiple unmanned aerial vehicles for optimal wireless coverage,'' \emph{IEEE Commun. Lett.}, vol. 20, no. 8, pp. 1647--1650, Aug. 2016.
    
   	\bibitem{Yang2018Joint}
    Z. Yang, C. Pan, M. Shikh-Bahaei, W. Xu, M. Chen, M. Elkashlan, and A. Nallanathan, ``Joint altitude, beamwidth, location, and bandwidth optimization for UAV-enabled communications,'' \emph{IEEE Commun. Lett.}, vol. 22, no. 8, pp. 1716--1719, Aug. 2018.
    
    \bibitem{Mozaffari2017Mobile}
    M. Mozaffari, W. Saad, M. Bennis, and M. Debbah, ``Mobile Internet of Things: Can UAVs provide an energy-efficient mobile architecture?'' in \emph{Proc. IEEE GLOBECOM}, Washington, DC, USA, Dec. 2016, pp. 1--6.
	
	\bibitem{Zeng2017Energy}
	Y. Zeng and R. Zhang, ``Energy-efficient UAV communication with trajectory optimization,'' \emph{IEEE Trans. Wireless Commun.}, vol. 16, no. 6, pp. 3747--3760, Jun. 2017.
	
	\bibitem{Cai2017Low}
	X. Cai, \textit{et al}, ``Low altitude UAV propagation channel modelling,'' in \emph{Proc. 11th Eur. Conf. Antennas Propag. (EuCAP17)}, Paris, France, Mar. 2017, pp. 1443--1447.
	
	\bibitem{LinStriking}
	X. Lin, G. Su, B. Chen, H. Wang, and M. Dai, ``Striking a balance between system throughput and energy efficiency for UAV-IoT systems,'' \emph{IEEE Internet Things J.}, to be published.
	
	\bibitem{Fouda2019Interference}
	A. Fouda, \textit{et al}, ``Interference management in UAV-assisted integrated access and backhaul cellular networks,'' \emph{IEEE Access}, vol. 7, pp. 104553--104566, 2019.
	
	\bibitem{Mei2018Uplink}
	W. Mei, Q. Wu, and R. Zhang, ``Cellular-connected UAV: Uplink association, power control and interference coordination,'' in \emph{Proc. IEEE Global Commun. Conf. (GLOBECOM)}, Abu Dhabi, UAE, Dec. 2018, pp. 206--212.
	
	\bibitem{Cheng2018UAV}
	F. Cheng, S. Zhang, Z. Li, Y. Chen, N. Zhao, R. Yu, and V. C. M. Leung, ``UAV trajectory optimization for data offloading at the edge of multiple cells,'' \emph{IEEE Trans. Veh. Technol.}, vol. 67, no. 7, pp. 6732--6736, Jul. 2018.
	
	
	\bibitem{Zhang2019IoT}
	Q. Zhang, M. Jiang, Z. Feng, W. Li, W. Zhang, and M. Pan, ``IoT enabled UAV: Network architecture and routing algorithm,'' \emph{IEEE Internet of Things J.}, vol. 6, no. 2, pp. 3727--3742, Apr. 2019.
	
	\bibitem{Yuan2019Joint}
	Q. Yuan, Y. Hu, C. Wang, and Y. Li, ``Joint 3D beamforming and trajectory design for UAV-enabled mobile relaying system,'' \emph{IEEE Access}, vol. 7, pp. 26488--26496, Feb. 2019.
	
	\bibitem{Zhao2019Caching}
	N. Zhao, F. R. Yu, L. Fan, Y. Chen, J. Tang, A. Nallanathan, and V. C. M. Leung, ``Caching unmanned aerial vehicle-enabled small-cell networks: Employing energy-efficient methods that store and retrieve popular content,'' \emph{IEEE Veh. Technol. Mag.}, vol. 14, no. 1, pp. 71--79, Mar. 2019.
	
	\bibitem{Mozaffari2019A}
    M. Mozaffari, W. Saad, M. Bennis, Y. N. Nam, and M. Debbah, ``A tutorial on UAVs for wireless networks: Applications, challenges and open problems,'' \emph{IEEE Commun. Survey Tuts.}, vol. 21, no. 3, pp. 2334--2360, Mar. 2019.	

    \bibitem{Filippone2006Flight}
    A. Filippone, \emph{Flight Performance of Fixed and Rotary Wing Aircraft} (Butterworth-Heinemann). Washington, DC, USA: AIAA, 2006.
	
	\bibitem{Goldsmith2005Wireless}
	A. Goldsmith, \emph{Wireless Communications.} Cambridge, U.K.: Cambridge Univ. Press, 2005.
	
	\bibitem{Boyd2004Convex}
	S. Boyd and L. Vandenberghe, \emph{Convex Optimization.} Cambridge, U.K.: Cambridge Univ. Press, 2004.
	
    \bibitem{Kobayashi2006An}	
    M. Kobayashi and G. Caire, ``An iterative water-filling algorithm for maximum weighted sum-rate of Gaussian MIMO-BC,'' \emph{IEEE J. Sel. Areas Commun.}, vol. 24, no. 8, pp. 1640--1646, Aug. 2006.

    \bibitem{Chen2017Caching}
    M. Chen, M. Mozaffari, W. Saad, C. Yin, M. Debbah, and C. S. Hong, ``Caching in the sky: Proactive deployment of cache-enabled unmanned aerial vehicles for optimized quality-of-experience,'' \emph{IEEE J. Sel. Areas Commun.}, vol. 35, no. 5, pp. 1046--1061, May 2017.
    
    \bibitem{Zhao2019UAV}
    N. Zhao, W. Lu, M. Sheng, Y. Chen, J. Tang, F. Yu, and K. Wong, ``UAV-assisted emergency networks in disasters,'' \emph{IEEE Wireless Commun.}, vol. 26, no. 1, pp. 45--51, Feb. 2019.
    

	
	
	
	
\end{thebibliography}
\end{document}